\documentclass[11pt]{article}
\usepackage{amsmath,amsthm,amsfonts,amssymb,amscd}
\usepackage{enumitem}
\usepackage[margin=1in]{geometry}
\usepackage{hyperref}
\usepackage{mathtools}
\usepackage{tikz}

\newtheorem{theorem}{Theorem}[section]
\newtheorem{lemma}[theorem]{Lemma}
\newtheorem{corollary}[theorem]{Corollary}
\newtheorem{definition}[theorem]{Definition}
\newtheorem{proposition}[theorem]{Proposition}

\newcommand\C{{\mathbb{C}}}

\newcommand\F{{\mathbb{F}}}

\newcommand\irr{{\mathop\textup{Irr}}}

\newcommand\remove[1]{{}}

\newcommand{\iso}{\cong}

\title{Fast generalized DFTs for all finite
  groups}
\author{Chris Umans\thanks{Supported by NSF grant CCF-1815607 and a Simons Foundation Investigator grant.}\\ Caltech}

\begin{document}
\maketitle

\date
\abstract{For any finite group $G$, we give an arithmetic algorithm to compute generalized Discrete Fourier Transforms
  (DFTs) with respect to $G$, using $O(|G|^{\omega/2 + \epsilon})$
  operations, for any $\epsilon > 0$. Here, $\omega$ is
  the exponent of matrix multiplication.}

\section{Introduction}

For a finite group $G$, let $\irr(G)$ denote a complete set of
irreducible representations of $G$. A {\em generalized DFT with
  respect to $G$} is a map from a group algebra element $\alpha \in \C[G]$ (which is
a vector of $|G|$ complex numbers), to the following linear
combination of irreducible representations:
\[\sum_{g \in G} \alpha_g \bigoplus_{\rho \in \irr(G)} \rho(g).\]
It is unique once one fixes a basis for each $\rho$; one usually seeks
algorithms that
work for arbitrary chosen bases. 
We typically speak of the complexity of computing this map in the (non-uniform) arithmetic
circuit model and do not concern ourselves with {\em finding} the
irreducible representations. The trivial algorithm thus requires $O(|G|^2)$
operations, since we are summing $|G|$ block-diagonal matrices,
each with $|G|$ entries in the blocks.

Fast algorithms for the DFT with respect to cyclic groups are
well-known and are attributed to Cooley and Tukey in 1965 \cite{CT}, although the
ideas likely date to Gauss. Beth in 1984 \cite{beth}, together with
Clausen \cite{clausenfast},
initiated the study of generalized DFTs, the ``generalized''
terminology signalling that the underlying group may be any group. A
central goal since that time has been to obtain fast algorithms for
generalized DFTs with respect to arbitary underlying groups. One may hope
for ``nearly-linear'' time algorithms, meaning that they use a number
of operations that is upper-bounded by $c_\epsilon|G|^{1 + \epsilon}$ for 
 universal constants $c_\epsilon$ and arbitary $\epsilon > 0$. Such ``exponent
 one'' algorithms are known for certain families of groups: abelian
 groups, supersolvable groups \cite{Baum}, and symmetric and alternating
 groups \cite{clausenfast}. Algorithms for generalized DFTs often find
 themselves manipulating matrices, so it is not surprising that they
 require a number of operations that depends on $\omega$, the exponent
 of matrix multiplication. Thus we view algorithms that achieve exponent
 one conditioned on $\omega = 2$ as being ``nearly as good'' as unconditional
 exponent one algorithms. Such algorithms are known for solvable
 groups \cite{beth, CB}, and with the recent breakthrough of \cite{HU18}, for
 linear groups; these algorithms achieve exponent $\omega/2$. 

In this paper we realize the main goal of the area, obtaining exponent
$\omega/2$ for all finite groups $G$. The previous best exponent that
applies to all finite groups was obtained by \cite{HU18}; it depends in
a somewhat complicated way on $\omega$, but it is at best $\sqrt{2}$
(when $\omega = 2$). Before that, the best known exponent was $1 +
\omega/4$ (which is at best $3/2$ when $\omega =2$),
and this dates back to the original work of Beth and Clausen.

\subsection{Past and related work}

A good description of past work in this area can be found in Section 13.5
of \cite{BCS}. The first algorithm generalizing beyond the
abelian case is due to Beth in 1984 \cite{beth}; this algorithm is
described in Section \ref{sec:single-subgroup} in a form often
credited jointly to Beth and Clausen. Three other milestones are the
$O(|G|\log|G|)$ algorithm for supersolvable groups due to Baum
\cite{Baum}, the $O(|G|\log^3 |G|)$ algorithm for the symmetric group
due to Clausen \cite{clausenfast} (see also \cite{Maslen} for a recent
improvement), and the $O(|G|^{\omega/2 +
  \epsilon})$ algorithms for linear groups obtained by Hsu and Umans,
which is described in Section \ref{sec:double-subgroup}. Wreath products were studied by Rockmore \cite{wreath} who obtained exponent one algorithms in certain cases. 

In the 1990s, Maslen, Rockmore, and coauthors developed the so-called ``separation
of variables'' approach \cite{MRsurvey2}, which relies on non-trivial decompositions along
chains of subgroups via {\em Bratteli diagrams} and detailed
knowledge of the representation theory of the underlying groups. There
is a rather large body of literature on this approach and it has been
applied to a wide variety of group algebras and more general algebraic
objects. For a fuller description of this approach and the results
obtained, the reader is referred to the surveys \cite{MRsurvey, Rrecent}, and the most
recent paper in this line of work \cite{MRW}.

\section{Preliminaries} 

Throughout this paper we will use the phrase 

\smallskip
\centerline{``generalized DFTs with respect to $G$ can be computed using $O(|G|^{\alpha + \epsilon})$ operations, for all
$\epsilon > 0$''} 
\smallskip

\noindent where $G$ is a finite group and $\alpha \ge 1$ is a
real number. We mean by this that
there are {\em universal} constants $c_\epsilon$ independent of the
group $G$ under consideration so that for each $\epsilon > 0$, the
operation count is at most $c_\epsilon|G|^{\alpha + \epsilon}$. Such
an algorithm will be referred to as an ``exponent $\alpha$''
algorithm. This comports with the precise definition of the exponent
of matrix multiplication, $\omega$: that there are universal
constants $b_\epsilon$ for which $n \times n$ matrix
multiplication can be performed using at most $b_\epsilon n^{\omega +
  \epsilon}$ operations, for each $\epsilon > 0$. Indeed we will often
report our algorithms' operation counts in terms of $\omega$. In such
cases matrix multiplication is always used as a black box, so,
for example, an
operation count of $O(|G|^{\omega/2})$ should be interpreted to mean: if one uses a fast matrix multiplication algorithm with exponent $\alpha$
(which may range from $2$ to $3$), then the operation count is
$O(|G|^{\alpha/2}$). In particular, in real implementations, one
might well use standard matrix multiplication and plug in $3$ for
$\omega$ in the operation count bound.  

We use $\irr(G)$ to denote the complete set
of irreducible representations of $G$ being used for the DFT at
hand. In the presentation to follow, we assume the underlying field is
$\C$; however our algorithms work over any field $\F_{p^k}$ whose
characteristic $p$ does not
divide the order of the group, and for which $k$ is sufficiently large
for $\F_{p^k}$ to represent a complete set of irreducibles. 

We use $I_n$ to denote the $n \times n$ identity matrix. The following
is an important general observation (see, e.g., Lemma
4.3.1 in \cite{HJ}):  
\begin{proposition}
\label{prop:mat-mult}
If $A$ is an $n_1 \times n_2$ matrix, $B$ is an $n_2 \times n_3$ matrix, and $C$ is an
$n_3 \times n_4$ matrix, then the entries of 
the product matrix $ABC$ are exactly the entries of the vector
obtained by multiplying $A \otimes C^T$
(which is an $n_1n_4 \times n_2n_3$ matrix) by $B$ viewed as an
$n_2n_3$-vector, and denoted $\mbox{vec}(B)$.
\end{proposition}

\subsection{Basic representation theory}
A {\em representation} of group $G$ is a homomorphism $\rho$ from $G$ into the
group of invertible $d \times d$ matrices. Representation $\rho$ naturally
specifies an action of $G$ on $\C^d$; representation $\rho$ is thus
said to have {\em dimension} $\dim(\rho) = d$.  A representation is
{\em irreducible} if the action on $\C^d$ has no $G$-invariant
subspace. Two representations of
the same dimension $d$,
$\rho_1$ and $\rho_2$, are
{\em equivalent} (written $\rho_1 \iso \rho_2$) if they are the same up to a change of basis; i.e.,
$\rho_1(g) = T\rho_2(g)T^{-1}$ for some invertible $d \times d$ matrix
$T$. The classical Maschke's Theorem implies that every
representation $\rho_0$ of $G$ {\em breaks up} into the direct sum of irreducible
representations; i.e. there is an invertible matrix $T$ and a multiset
$S \subseteq \irr(G)$, for which
\[T\rho_0(g)T^{-1} = \bigoplus_{\rho \in S} \rho(g).\]

Given a subgroup $H \subseteq G$ one can obtain from any
representation $\rho \in \irr(G)$ a representation 
$\mbox{Res}^G_H(\rho)$ (the {\em restriction} of $\rho$ to $H$), which
is a representation of $H$, simply by restricting the domain of $\rho$
to $H$.  One can also obtain from any representation $\sigma \in
\irr(H)$, a representation of $G$ called the {\em induced}
representation $\mbox{Ind}^{G}_H(\rho)$, which has dimension
$\dim(\sigma)|G|/|H|$. We will not need to work directly with induced
representations, but we will use a fundamental fact called {\em
  Frobenius reciprocity}. Given $\rho \in \irr(G)$ and $\sigma \in
\irr(H)$, Frobenius reciprocity states that the number of times
$\sigma$ appears in the restriction $\mbox{Res}^G_H(\rho)$ equals the
number of times $\rho$ appears in the induced representation
$\mbox{Ind}^G_H(\sigma)$.  

A basic fact is that $\sum_{\rho \in \irr(G)} \dim(\rho)^2 = |G|$,
which implies that for all $\rho \in \irr(G)$, we have $\dim(\rho) \le
|G|^{1/2}$. This can be used to prove the following
inequality, which we use repeatedly:
\begin{proposition}
\label{prop:prelim-inequality}
For any real number $\alpha \ge 2$, we have
\[\sum_{\rho \in \irr(G)} \dim(\rho)^{\alpha} \le |G|^{\alpha/2}.\]
\end{proposition} 
\begin{proof}
Set $\rho_{\max}$ to be an irrep of largest dimension. We have
\[\sum_{\rho \in \irr(G)} \dim(\rho)^{\alpha} \le 
\dim(\rho_{\max})^{\alpha - 2}\sum_{\rho \in \irr(G)} \dim(\rho)^2 =
\dim(\rho_{\max})^{\alpha - 2}|G| \le  |G|^{\alpha/2},\]
where the last inequality used the fact that $\dim(\rho_{\max}) \le |G|^{1/2}$.
\end{proof}

\subsection{Basic Clifford theory}
\label{sec:clifford}

Clifford theory describes the way the irreducible representations of a
group $H$ break up when restricted to a {\em normal} subgroup $N$,
which is a particularly well-structured and well-understood scenario. 

Elements of $H$ act on the set $\irr(N)$ as follows:\[(h\cdot \lambda)(n) =
\lambda(hnh^{-1}), \]
 for  $\lambda \in \irr(N)$. 
Let ${\cal O}_1, \ldots, {\cal O}_\ell$ be the orbits of this $H$-action
on $\irr(N)$. Clifford theory states for each $\sigma \in \irr(H)$,
there is a positive integer $e_\sigma$ and an index $i_\sigma$ for which 
the restriction $\mbox{Res}_N^H(\sigma)$ is equivalent to 
\[e_\sigma \bigoplus_{\lambda \in {\cal O}_{i_{\sigma}} } \lambda.\] 
In particular, this implies that all $\lambda \in \irr(N)$ that occur in the
restriction have the same dimension, $d_\sigma$, and multiplicity, $e_\sigma$, and that
$\dim(\sigma) = d_\sigma e_\sigma|{\cal O}_{i_\sigma}|$.

We can also define the following subsets, which partition $\irr(H)$:
\[S_\ell = \{\sigma \in \irr(H) : \mbox{the irreps in ${\cal O}_\ell$ occur
  in $\sigma$}\} = \{\sigma \in \irr(H) : i_\sigma = \ell\}.\]



We will need the following proposition:
\begin{proposition}
\label{prop:equality}
For a finite group $H$ and normal subgroup $N$, and sets $S_\ell$ as
defined above, the following holds for each $\ell$:
\[\sum_{\sigma \in S_\ell} \dim(\sigma)e_\sigma/d_\sigma =
  |H/N|.\] 
\end{proposition}
\begin{proof}
Fix $\lambda \in {\cal O}_\ell$, and note that the induced representation  $\mbox{Ind}_N^H(\lambda)$
has dimension $\dim(\lambda)|H/N|$.  Let $m_{\sigma, \lambda}$ be the
number of times $\sigma \in \irr(H)$ occurs in $\mbox{Ind}_N^H(\lambda)$. Then we have
\[\sum_{\sigma \in \irr(H)} \dim(\sigma)m_{\sigma, \lambda} =
  \dim(\lambda)|H/N|.\] 
By Frobenius reciprocity, $m_{\sigma, \lambda}$ equals the
number times $\lambda$ occurs in  $\mbox{Res}_N^H(\sigma)$. Thus the summand $\dim(\sigma)m_{\sigma,
  \lambda}$ equals $\dim(\sigma)e_\sigma$, whenever
$m_{\sigma, \lambda} \ne 0$ (and zero otherwise). The proposition follows. 
\end{proof}


\subsection{Generalized DFTs and inverse generalized DFTs}

We assume by default that we
are computing generalized DFTs with respect to an arbitary chosen basis for each $\rho \in
\irr(G)$. Sometimes we need to refer to the special basis in the
following definition:
\begin{definition}
\label{def:h-adapted}
Let $H$ be a subgroup of $G$. An {\em $H$-adapted basis} is a basis
for each $\rho \in \irr(G)$, so that the restriction of $\rho$ to $H$ respects the direct
sum decomposition into irreps of $H$.
\end{definition}
In concrete terms, this implies that for each $\rho \in \irr(G)$, while for general $g \in G$,
$\rho(g)$ is a $\dim(\rho) \times \dim(\rho)$ matrix, for $g \in H$,
$\rho(g)$ is a block-diagonal matrix with block sizes coming from the
set $\{\dim(\sigma): \sigma \in \irr(H)\}$. An $H$-adapted basis
always exists.

A general trick that we will rely on is that if one can compute generalized
DFTs with respect to $G$ for an input $\alpha$ supported on a subset $S \subseteq G$, then with an
additional multiplicative factor of roughly $|G|/|S|$, one can
compute generalized DFTs with respect to $G$. 

\begin{theorem}
\label{thm:trick}
Fix a finite group $G$ and a subset $S \subseteq G$, and suppose a
generalized DFT with respect to $G$ can be
computed in $m$ operations, for inputs $\alpha$ supported on $S$. Then
generalized DFTs with respect to $G$ can be computed using \[O(m + |G|^{\omega/2 + \epsilon})\cdot \frac{|G|
  \log |G|}{|S|}\] operations,
for any $\epsilon > 0$. 
\end{theorem}
\begin{proof}
First observe that by multiplying by $\oplus_{\rho \in
  \irr(G)} \rho(g)$ we can compute a generalized DFT supported on
$Sg$, for an additive extra cost of \[\sum_{\rho \in \irr(G)}
O(\dim(\rho)^{\omega + \epsilon})\] operations, for all $\epsilon > 0$, and by applying
Proposition \ref{prop:prelim-inequality} with $\alpha = \omega +
\epsilon$ this is at most $O(|G|^{\omega/2 + \epsilon})$.
A probablistic argument shows that $|G|\log
|G|/|S|$ different translates $g$ suffice to cover $G$, so we need only repeat the
DFT supported on $Sg$ translated by each such $g$, and sum the
resulting DFTs. 
\end{proof}

The {\em inverse generalized DFT} maps a collection of matrices $M^{\sigma} \in
\C^{\dim(\sigma) \times \dim(\sigma)}$, one for each $\sigma \in
\irr(G)$, to the vector $\alpha$ for which 
\[\sum_{g \in G}\alpha_g \bigoplus_{\sigma \in \irr(G)} \rho(G) =
  \bigoplus_{\sigma \in \irr(G)} M^{\sigma}.\] 
In the arithmetic circuit model, the inverse DFT can be computed
efficiently if the DFT can:
\begin{theorem}[Baum, Clausen; Cor. 13.40 in \cite{BCS}]
Fix a generalized DFT with respect to finite group $G$ and suppose it
can be computed in $m$ operations. Then the
inverse DFT with respect to $G$ (and the same basis), can be computed
in at most $m + |G|$ operations. 
\end{theorem}

\section{General strategy: reduction to subgroups} 

One way to organize the main algorithmic ideas in the quest for a fast
DFT for all finite groups is according to the subgroup
structure they exploit.  The algorithms themselves are recursive, with
the main content of the algorithm being the reduction to smaller
instances: DFTs over subgroups of the original group. When aiming for generalized DFTs for all finite
groups, such a reduction is paired with a group-theoretic structural
result, which guarantees the existence of certain subgroups that are used by the
reduction.

In the exposition below, it is helpful to assume that $\omega = 2$
and seek an ``exponent 1'' algorithm under this assumption (in general, the exponent achieved
will be a function of $\omega$, and in our main result this function
is $\omega/2$). By the term {\em overhead} we mean the
extra multiplicative factor in the operation count of the reduction,
beyond the nearly-linear operation count that would be necessary for an
exponent 1 algorithm.

\subsection{The single subgroup reduction}
\label{sec:single-subgroup}
The seminal Beth-Clausen algorithm
reduces computing a DFT over a group $G$ to computing several DFTs over a
subgroup $H$ of $G$. We call this the ``single subgroup
reduction''. Roughly speaking, the overhead in this reduction is
proportional to the
index of $H$ in $G$. The companion structural result is Lev's Theorem \cite{lev},
which shows that every finite group $G$ (except cyclic of prime order
which can be handled separately) has a
subgroup of order at least $\sqrt{G}$ (and this is tight, hence the
overhead is $\sqrt{|G|}$ in the worst case). As noted in the introduction, this
reduction together with Lev's Theorem implies exponent $3/2$ (assuming
$\omega =2$) for all
finite groups. 

Here is a
more detailed description, together with results we will need later. Let $H$ be a subgroup of $G$ and let $X$ be a set of distinct coset
representatives. We first compute several $H$-DFTs, one for each $x \in X$:
\[s_x = \sum_{h \in H} \alpha_{hx} \bigoplus_{\sigma \in \irr(H)}
\sigma(h)\]
and by using an {\em $H$-adapted} basis (Definition \ref{def:h-adapted}), we can lift
each $s_x$ to 
\[\overline{s_x} = \sum_{h \in H} \alpha_{hx} \bigoplus_{\rho \in \irr(G)}
\rho(h)\]
by just copying entries (which is free of cost in the arithmetic
model). Then to complete the DFT we need to compute
\[\sum_{x \in X} \overline{s_x} \bigoplus_{\rho \in \irr(G)} \rho(x).\]
Generically, this final computation requires an overhead proportional to $|X| =
[G:H]$, even when just considering the outermost summation. See Corollary 4 in \cite{HU-full} for the details to complete this sketch, yielding
the following:
\begin{theorem}[single subgroup reduction]
\label{thm:single-subgroup}
Let $G$ be a finite group and let $H$ be a subgroup. Then we can
compute a generalized DFT with respect to $G$ at a cost of $[G : H]$
many $H$-DFTs plus $O([G : H]|G|^{\omega/2 + \epsilon})$ operations,  for all $\epsilon > 0$.
\end{theorem}

In the special case that $H$ is normal in $G$ and $G/H$ is cyclic of
prime order, the overhead of $[G:H]$ can be avoided, by using knowledge about
the way representations $\sigma \in \irr(H)$ extend to $\rho \in \irr(G)$. This
insight is the basis for the Beth-Clausen algorithm for solvable
groups. We need it here to handle the case of $G/H$ cyclic of prime
order, which is the single exceptional case not handled by our main
reduction. The following theorem can be inferred from the proof of Theorem 7.7 in
Clausen and Baum's monograph \cite{CB}:

\begin{theorem}[Clausen, Baum \cite{CB}]
\label{thm:prime-index}
Let $H$ be a normal subgroup of $G$ with prime index $p$. We can
compute a generalized DFT with respect to $G$ and an $H$-adapted basis, at a cost
of $p$ many $H$-DFTs plus 
\[O(p \log p)\cdot \sum_{\sigma \in \irr(H)} \dim(\sigma)^{\omega +
    \epsilon}\]
operations,  for all $\epsilon > 0$.
\end{theorem}

For our purposes the following slightly coarser bound suffices, which
accomodates an arbitary basis change (hence obviating the need for an
$H$-adapted basis):
\begin{corollary}
\label{cor:prime-index}
Let $H$ be a normal subgroup of $G$ with prime index $p$. Generalized
DFTs with respect to $G$ can be computed at a cost
of $p$ many $H$-DFTs plus $O(|G|^{\omega/2 + \epsilon})$
operations, for all $\epsilon > 0$. 
\end{corollary}

\begin{proof}
Applying Proposition \ref{prop:prelim-inequality} to Theorem
\ref{thm:prime-index} with $\alpha = \omega + \epsilon$ yields an
operation count of $O(p\log p)|H|^{\omega + \epsilon/2}$, which is at
most $O(|G|^{\omega/2+\epsilon})$. Performing
an arbitary basis change costs 
\[\sum_{\rho \in \irr(G)} O(\dim(\rho)^{\omega +
    \epsilon})\]
operations which is again at most $O(|G|^{\omega/2+\epsilon})$ by
Proposition \ref{prop:prelim-inequality}. 
\end{proof}

\subsection{The double subgroup reduction}
\label{sec:double-subgroup}
Recently, Hsu and Umans proposed a ``double subgroup
reduction'' \cite{HU18} which reduces computing a DFT over a group $G$
to computing
several DFTs over two subgroups, $H$ and $K$. This reduction is
especially effective for linear groups (see \cite{HU18}). Roughly speaking,
the overhead in this reduction is proportional to $|G|/|HK|$ and $|H
\cap K|$. The companion structural result shows that every finite
group $G$ (except $p$-groups which can be handled separately) has two
proper subgroups $H$ and $K$ for which $|G|/|HK|$ is negligible. However, $|H
\cap K|$ might still be large, which is the one thing standing in the
way of deriving an ``exponent $\omega/2$'' algorithm from this reduction.

To illustrate the bottleneck in this reduction, we describe it in more
detail. Let $H, K$ be subgroups of $G$ and assume $|G|/|HK|$ is negligible.
We first compute an intermediate representation \[\sum_{g = hk \in HK}\alpha_{g} \bigoplus_{\substack{\sigma \in
  \irr(H)\\ \tau \in \irr(K)}} \sigma(h) \otimes \tau(k)\] in two
steps (and then lift it to a $G$-DFT). The first of the two steps is
to 
compute at most $[G:H]$ many $H$-DFTs, yielding, for each $k \in K' \subseteq
K$ (where $K'$ is a set of distinct coset representatives of $H$ in $G$):
\[s_k = \sum_{h \in H} \alpha_{hk} \bigoplus_{\sigma \in \irr(H)}
\sigma(h).\]
The second step is as
follows: for each {\em entry} of the block-diagonal matrix $s_k$, we use this entry (as $k$
varies) as the data for a $K$-DFT. There are $\sum_{\sigma \in
  \irr(H)} \dim(\sigma)^2 = |H|$ such entries in
general. Thus the second step entails $|H|$ many $K$-DFTs, and this
represents the key bottleneck. Note that when $|G|/|HK|$ is negligible, $|H||K|$ is approximately
$|G||H\cap K|$, and this explains the overhead of roughly $|H
\cap K|$ which prevents obtaining an ``exponent $\omega/2$'' algorithm from this
reduction. For completeness we record the main theorem of \cite{HU-full}
here:
\begin{theorem}[Theorem 12 in \cite{HU-full}]
Let $G$ be a finite group and let $H,K$ be subgroups. Then we can compute generalized DFTs with respect to $G$ at the cost of $|H|$ many
$K$-DFTS, $|K|$ many $H$-DFTs, plus \[O(|G|^{\omega/2 + \epsilon}
  +(|H||K|)^{\omega/2 + \epsilon} )\]
operations, all repeated $O(\frac{|G|\log |G|}{|HK|})$ times, for all $\epsilon > 0$. 
\end{theorem}
Our main innovation, described in the next section, is a
way to overcome the bottleneck: when $H \cap K = N$ is a normal
subgroup of $G$, we are able to rewrite each $s_k$ as a sum of $|N|$ matrices
with special structure: effectively, there are only
$|H/N|$ many non-zero ``entries'' for which we need to compute a
$K$-DFT, and as we will show, this exactly removes the overhead factor.

\subsection{The triple subgroup reduction}

In this section we give our main new result.  We devise a ``triple subgroup reduction'' which reduces
computing a DFT over $G$ to computing several DFTs over two subgroups,
$H$ and $K$, {\em and} several inverse DFTs over the intersection $N =
H \cap K$, when $N$ is normal in $G$. Roughly speaking, the overhead is
proportional to $|G|/|HK|$.
The companion structural result (Theorem \ref{thm:triple-subgroup-structure}) shows that for every
finite group $G$, if $N$ is a
maximal normal subgroup in $G$ then (except for the case of $|G/N|$
cyclic of prime order, which can be handled separately) there exist
two proper subgroups $H$ and
$K$ with $H \cap K = N$, such that $|G|/|HK|$ is negligible. This
is the key to the claimed exponent $\omega/2$ algorithm.

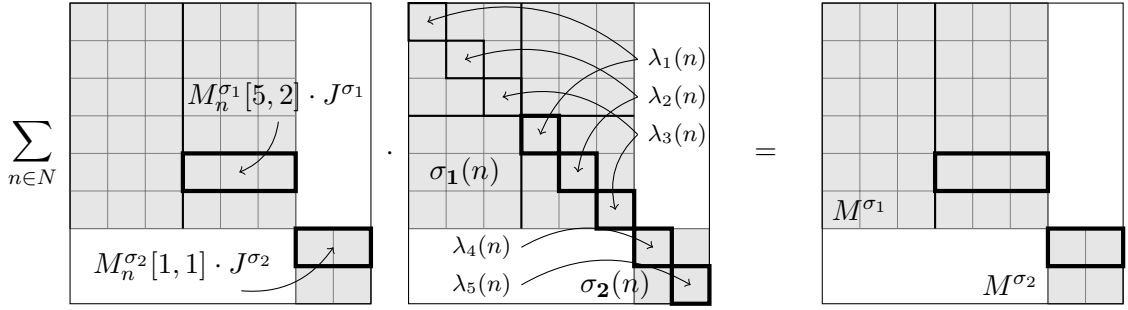
\begin{figure}[t]
\begin{center}
\begin{tikzpicture}[scale = 0.5]
\node at (0,5) {$\displaystyle \sum_{n\in N}$};

\draw [fill=black, fill opacity = 0.1] (1,3) rectangle (7,9);
\draw[help lines] (1,3) grid (7,9);
\draw [fill=black, fill opacity = 0.1] (7,1) rectangle (9,3);
\draw[help lines] (7,1) grid (9,3);
\draw (1,1) rectangle (9,9);
\draw [ultra thick] (4,4) node (b) {} rectangle (7,5);
\coordinate (b1) at (5.5, 4.5);
\node at (6.5, 6.5) (a1) {$M_n^{\sigma_1}[5,2]\cdot J^{\sigma_1}$};
\draw [->] (a1) to[bend left] (b1);
\draw [ultra thick] (7,2) node (d) {} rectangle (9,3);
\coordinate (b2) at (8, 2.5);
\node at (4,2) (a2) {$M_n^{\sigma_2}[1,1]\cdot J^{\sigma_2}$};
\draw [->] (a2) to[bend right] (b2);
\draw [thick] (4,9) -- (4,3);

\node at (9.5,5) {$\cdot$};

\draw [fill=black, fill opacity = 0.1] (10,3) rectangle (16,9);
\draw[help lines] (10,3) grid (16,9);
\draw [fill=black, fill opacity = 0.1] (16,1) rectangle (18,3);
\draw[help lines] (16,1) grid (18,3);
\draw (10,1) rectangle (18,9);
\draw [thick] (13,9) -- (13,3);
\draw [thick] (10,6) -- (16,6);
\node at (11.5, 4.5) {$\mathbf{\sigma_1}(n)$};
\node at (15.5, 1.5) {$\mathbf{\sigma_2}(n)$};

\draw [ thick]  (10,9) rectangle (11,8);
\coordinate (r1) at (10.5, 8.5);
\node at (17.15,7.5) (l1) {\footnotesize $\lambda_1(n)$};
\draw [<-] (r1) to[bend left] (l1.west); 
\draw [ thick]  (11,8)  rectangle (12,7);
\coordinate (r2) at (11.5, 7.5);
\node at (17.15, 6.5) (l2){\footnotesize $\lambda_2(n)$};
\draw [<-] (r2) to[bend left] (l2.west); 
\draw [ thick]  (12,7)  rectangle (13,6);
\coordinate (r3) at (12.5, 6.5);
\node at (17.15, 5.5) (l3){\footnotesize $\lambda_3(n)$};
\draw [<-] (r3) to[bend left] (l3.west); 
\draw [ultra thick]  (13,6)  rectangle (14,5);
\coordinate (r4) at (13.5, 5.5);
\draw [<-] (r4) to[bend left] (l1.west); 
\draw [ultra thick]  (14,5)  rectangle (15,4);
\coordinate (r5) at (14.5, 4.5);
\draw [<-] (r5) to[bend left] (l2.west); 
\draw [ultra thick]  (15,4) rectangle (16,3);
\coordinate (r6) at (15.5, 3.5);
\draw [<-] (r6) to[bend left] (l3.west); 

\draw [ultra thick]  (16,3) rectangle (17,2);
\coordinate (r7) at (16.5, 2.5);
\node at (13, 2.5) [anchor=east] (l7) {\footnotesize $\lambda_4(n)$};
\draw [<-] (r7) to[bend right] (l7.east); 
\draw [ultra thick]  (17,2) rectangle (18,1);
\coordinate (r8) at (17.5, 1.5);
\node at (13, 1.5) [anchor=east] (l8) {\footnotesize $\lambda_5(n)$};
\draw [<-] (r8) to[bend right] (l8.east); 

\node at (19.5,5) {$=$};

\draw [fill=black, fill opacity = 0.1] (21,3) rectangle (27,9);
\draw[help lines] (21,3) grid (27,9);
\draw [fill=black, fill opacity = 0.1] (27,1) rectangle (29,3);
\draw[help lines] (27,1) grid (29,3);
\draw (21,1) rectangle (29,9);
\draw [thick] (24,3) -- (24,9);
\draw [ultra thick] (24,4) rectangle (27,5);
\draw [ultra thick] (27,3) rectangle (29,2);
\node at (22, 3.5) {${M^{\sigma_1}}$};
\node at (26, 1.5) {${M^{\sigma_2}}$};

\end{tikzpicture}
\end{center}
\caption{Illustration of the 
  proof of Theorem \ref{thm:sparse}. In this example $\irr(H) =
  \{\sigma_1, \sigma_2\}$, $\irr(N) = \{\lambda_1,
  \lambda_2,\lambda_3, \lambda_4, \lambda_5\}$; the orbits are ${\cal
    O}_1 = \{\lambda_1, \lambda_2, \lambda_3\}$ and ${\cal O}_2 =
  \{\lambda_4, \lambda_5\}$; $S_1 = \{\sigma_1\}$ and $S_2 =
  \{\sigma_2\}$; and the multiplicities are $e_{\sigma_1} = 2$ and
  $e_{\sigma_2} = 1$. In the figure, we highlight the parts of the
  matrices that give rise to the system of equations solved with a
  single inverse $N$-DFT, corresponding to the value $a = f_1(\sigma_1,
  5,2) = f_2(\sigma_2,1,1)$. This inverse $N$-DFT with
  the highlighted blocks of $M^{\sigma_1}$ and $M^{\sigma_2}$ as input data
  yields the scalars $M_n^{\sigma_1}[5,2] = M_n^{\sigma_2}[1,1]$ that
  satisfy the simultaneous equations.}
\label{fig:inverse-DFT}
\end{figure}

Let $H$ be a group with normal subgroup $N$. The main technical theorem shows how to rewrite the output of an $H$-DFT as the
sum of $|N|$ matrices each of which only has ``$|H/N|$ degrees of
freedom''. In the following theorem we adopt the notation introduced
in Section \ref{sec:clifford}.
\begin{theorem}
\label{thm:sparse}
Let $H$ be a group and $N$ a normal subgroup. 
For every 
\[M = \bigoplus_{\sigma \in \irr(H)} M^\sigma \in \bigoplus_{\sigma \in
    \irr(H)} \C^{\dim(\sigma) \times \dim(\sigma)},\]
the following holds with respect to an $N$-adapted basis: there exist
matrices $M^\sigma_n \in \C^{\dim(\sigma)/d_\sigma \times e_\sigma}$ for which 
\[\sum_{n \in N} (M^{\sigma}_n \otimes J^\sigma) \cdot \sigma(n) = M^\sigma,\]
where $J^{\sigma}$ is the $d_\sigma \times \dim(\sigma)/e_\sigma $
matrix $(I_{d_\sigma} | I_{d_\sigma} | \cdots | I_{d_\sigma})$. Moreover, given injective functions $f_\ell$ from $\{(\sigma, i,
  j): \sigma \in S_\ell,  i \in
  [\dim(\sigma)/d_\sigma], j \in [e_\sigma]\}$ to $[r]$, the $M_n^\sigma$ can be taken
  to satisfy 
\[f_\ell(\sigma, i, j) = f_{\ell'}(\sigma', i', j') \;\;\;\;
  \Rightarrow \;\;\;\; 
  \forall n \; M_n^\sigma[i,j] = M_n^{\sigma'}[i',j'],\]
and these matrices $M^\sigma_n$ can be obtained from $M$ by computing $r$ inverse $N$-DFTs.
\end{theorem}
One should think of the functions $f_\ell$ as labeling the entries of
the $M_n^\sigma$ matrices for the $\sigma$ in a given $S_\ell$. This labeling is
then used to ensure that entries of $M_n^\sigma$ with $\sigma \in
S_\ell$ and the entries of $M_n^{\sigma'}$ with $\sigma' \in
S_{\ell'}$ are equal, if they have the same labels. In Section \ref{sec:labeling} we
will show how to choose this labeling so that the final ``lifting''
step of our algorithm can be efficiently computed. For now, we note
that Proposition \ref{prop:equality} implies that there {\em exist}
labellings $f_\ell$ with $r = |H/N|$, and indeed our actual choice of
$f_\ell$ in Section \ref{sec:labeling} will have $r = O(|H/N|\log|H/N|)$, which is
not much larger.   
\begin{proof}
Fix $\sigma \in \irr(H)$, and recall that there is a unique $S_\ell$
containing $\sigma$. Since we are using an $N$-adapted basis, $\sigma(n)$ has the form 
\[I_{e_\sigma} \otimes \bigoplus_{\lambda \in {\cal O}_\ell } \lambda(n),\]
and thus
\begin{equation}
\label{eq:structured}
\sum_{n \in N} (M^{\sigma}_n \otimes J^{\sigma}) \cdot \sigma(n) =
  \sum_{n \in N} M^{\sigma}_n \otimes (\lambda_1(n) | \lambda_2(n) |
  \cdots |\lambda_{|{\cal O}_\ell|}(n))
\end{equation}
where $\lambda_1, \ldots, \lambda_{|{\cal O}_\ell|}$ is an enumeration of
${\cal O}_\ell$.  Since these are pairwise inequivalent irreps, the span of 
\[\{ (\lambda_1(n) | \lambda_2(n) |
  \cdots |\lambda_{|{\cal O}_\ell|}(n)) : n \in N\}\]
is the full matrix algebra $\C^{d_\sigma \times \dim(\sigma)/e_\sigma}$. Hence we can choose the $M^{\sigma}_n$ so that
expression (\ref{eq:structured}) equals an arbitrary $M^{\sigma} \in
\C^{\dim(\sigma) \times \dim(\sigma)}$. 

In particular, for each $\sigma$, the $(i,j)$
entries of the $M^{\sigma}_n$ should satisfy
\begin{equation}
\label{eq:inverse-N-DFT}
\sum_{n \in N} M^{\sigma}_n[i,j] \left ( \begin{array}{l} \lambda_1(n) \\ \lambda_2(n) \\
  \vdots \\ \lambda_{|{\cal O}_\ell|}(n)\end{array}\right ) = 
\left (\begin{array}{l}M^\sigma[i, j\cdot |{\cal O}_\ell|] \\
  M^{\sigma}[i, j\cdot |{\cal O}_\ell| + 1]\\ \vdots \\
         M^{\sigma}[i, j\cdot |{\cal O}_\ell|+|{\cal O}_\ell|-1]\end{array}\right )
\end{equation}
where $M^\sigma$ on the right-hand-side is blocked into $d_\sigma
\times d_\sigma$ submatrices and indexed accordingly. Thus the values of a given entry of $M^{\sigma}_n$ as $n$ ranges
  over $N$, can be found in an inverse $N$-DFT with the appropriate 
  blocks of $M^{\sigma}$ as input data. 

Observe however that in general, ${\cal O}_\ell$ is a {\em proper}
subset of $\irr(H)$, and hence the aforementioned inverse $N$-DFT is
underdetermined; for example Equation (\ref{eq:inverse-N-DFT}) remains
satisfied if we require $\sum_{n \in N}M_n^\sigma[i,j] \lambda(n) = 0$
for all $\lambda \in \irr(H) \setminus {\cal O}_\ell$. 

Indeed, we can
{\em simultaneously} solve Equation (\ref{eq:inverse-N-DFT}) with
respect to several $\sigma \in \irr(H)$ via a single inverse
$N$-DFT, provided the associated orbits ${\cal O}_{i_\sigma}$ are
different. To prove the ``moreover'' part of the theorem statement,
then, we set up the following system of equations, for a given $a \in
[r]$: for each $\ell$ for which $f_\ell(\sigma, i,j) = a$ we
{\em simultaneously} require that Equation (\ref{eq:inverse-N-DFT}) holds with
respect to $\sigma, i, j$ (and note these are determined by $a$ since
$f_\ell$ is injective). Since the $S_\ell$ partition $\irr(H)$, 
selecting at most one $\sigma$ from each $S_\ell$ results in a system
that mentions each $\lambda \in \irr(N)$ at most once. Hence a single inverse $N$-DFT solves this system of
equations. See Figure \ref{fig:inverse-DFT}. We do this once for each $a \in [r]$, to produce the
matrices $M^\sigma_n$ from the original $M$, using $r$ inverse $N$-DFTs.
\end{proof}

\subsubsection{Choosing the labelings $f_\ell$}
\label{sec:labeling}

To make use of Theorem \ref{thm:sparse}, we need to define injective functions $f_\ell$ from \[\{(\sigma, i,
  j): \sigma \in S_\ell, i \in
  [\dim(\sigma)/d_\sigma], j \in [e_\sigma]\}\] to $[r]$. We identify the domain of
  $f_\ell$ with the entries of a block-diagonal matrix, with
  rectangular blocks of size $\dim(\sigma)/d_\sigma \times e_\sigma$,
  as $\sigma$ ranges over $S_\ell$. Recall that by Proposition
  \ref{prop:equality}, the total number of entries in these blocks is
  $|H/N|$.

\begin{figure}[t]
\begin{center}
\begin{tikzpicture}[scale = 0.5]

\begin{scope}[shift={(2,12)}]
\draw (1,1) rectangle (6,12);
\draw [fill=black, fill opacity = 0.1] (1,12) rectangle (2,8);
\draw[help lines] (1,12) grid (2,8);
\draw [fill=black, fill opacity = 0.1] (2,8) rectangle (3,5);
\draw[help lines] (2,8) grid (3,5);
\draw [fill=black, fill opacity = 0.1] (3,5) rectangle (4,3);
\draw[help lines] (3,5) grid (4,3);
\draw [fill=black, fill opacity = 0.1] (4,3) rectangle (6,1);
\draw[help lines] (4,3) grid (6,1);
\node at (6.5,9) {$\cdot$};
\draw (7,12) rectangle (8,11);
\node at (7.5,11.5) {$x_1$};
\draw (7,11) rectangle (8,10);
\node at (7.5,10.5) {$x_2$};
\draw (7,10) rectangle (8,9);
\node at (7.5,9.5) {$x_3$};
\draw (7,9) rectangle (8,7);
\node at (7.5,8.5) {$x_4'$};
\node at (7.5,7.5) {$x_4''$};
\node at (8.5,9) {$=$};
\draw (9,12) rectangle (10,8);
\node at (9.5,10) {$y_1$};
\draw (9,8) rectangle (10,5);
\node at (9.5,6.5) {$y_2$};
\draw (9,5) rectangle (10,3);
\node at (9.5,4) {$y_3$};
\draw (9,3) rectangle (10,1);
\node at (9.5,2) {$y_4$};
\node[anchor=west] at (10,2) {$(=y_4' + y_4'')$};

\draw [ultra thick]  (1,12)  rectangle (2,8);
\node at (1.5,10) {$1$};

\draw [ultra thick]  (2,8)  rectangle (3,5);
\node at (2.5,6.5) {$2$};

\draw [ultra thick]  (3,5)  rectangle (4,3);
\node at (3.5,4) {$3$};

\draw [ultra thick]  (4,3)  rectangle (5,1);
\node at (4.5,2) {$4$};

\draw [ultra thick]  (5,3)  rectangle (6,1);
\node at (5.5,2) {$5$};
\end{scope}

\begin{scope}[shift={(2,12)}]
\draw (17,3) rectangle (21,12);
\draw [fill=black, fill opacity = 0.1] (17,12) rectangle (19,8);
\draw[help lines] (17,12) grid (19,8);
\draw [fill=black, fill opacity = 0.1] (19,8) rectangle (20,5);
\draw[help lines] (19,8) grid (20,5);
\draw [fill=black, fill opacity = 0.1] (20,5) rectangle (21,3);
\draw[help lines] (20,5) grid (21,3);
\node at (21.5,10) {$\cdot$};
\draw (22,12) rectangle (23,10);
\node at (22.5,11) {$u_1$};
\draw (22,10) rectangle (23,9);
\node at (22.5,9.5) {$u_2$};
\draw (22,9) rectangle (23,8);
\node at (22.5,8.5) {$u_3$};
\node at (23.5,10) {$=$};
\draw (24,12) rectangle (25,8);
\node at (24.5,10) {$v_1$};
\draw (24,8) rectangle (25,5);
\node at (24.5,6.5) {$v_2$};
\draw (24,5) rectangle (25,3);
\node at (24.5,4) {$v_3$};

\end{scope}

\begin{scope}[shift={(1,0)}]
\draw (1,4) rectangle (9,12);
\draw [fill=black, fill opacity = 0.1] (1,12) rectangle (5,8);
\draw[help lines] (1,12) grid (5,8);
\draw [ultra thick]  (1,12)  rectangle (2,8);
\node at (1.5,10) {$1$};

\draw [ultra thick]  (2,12)  rectangle (3,9);
\node at (2.5,10.5) {$2$};

\draw [fill=black, fill opacity = 0.1] (5,8) rectangle (7,6);
\draw[help lines] (5,8) grid (7,6);
\draw [ultra thick]  (5,8)  rectangle (6,6);
\node at (5.5,7) {$3$};

\draw [ultra thick]  (6,8)  rectangle (7,6);
\node at (6.5,7) {$4$};

\draw [fill=black, fill opacity = 0.1] (7,6) rectangle (9,4);
\draw[help lines] (7,6) grid (9,4);
\draw [ultra thick]  (7,6)  rectangle (8,4);
\node at (7.5,5) {$5$};

\node at (9.5, 8) {$\cdot$};

\draw (10,4) rectangle (18,12);
\draw [fill=black, fill opacity = 0.1] (10,12) rectangle (14,8);
\draw [fill=black, fill opacity = 0.1] (14,8) rectangle (17,6);
\draw [fill=black, fill opacity = 0.1] (17,6) rectangle (18,4);

\draw (10,12) rectangle (11,11);
\node at (10.5,11.5) {$x_1$};
\draw (11,11) rectangle (12,10);
\node at (11.5,10.5) {$x_2$};
\draw (12,12) rectangle (13,10);
\node at (12.5,11) {$u_1$};
\draw (13,10) rectangle (14,9);
\node at (13.5,9.5) {$u_2$};

\draw (14,8) rectangle (15,7);
\node at (14.5,7.5) {$x_3$};
\draw (15,7) rectangle (16,6);
\node at (15.5,6.5) {$x_4'$};
\draw (16,8) rectangle (17,7);
\node at (16.5,7.5) {$u_3$};

\draw (17,6) rectangle (18,5);
\node at (17.5,5.5) {$x_4''$};

\node at (18.5, 8) {$=$};

\draw (19,4) rectangle (27,12);
\draw [fill=black, fill opacity = 0.1] (19,12) rectangle (23,8);
\draw [fill=black, fill opacity = 0.1] (23,8) rectangle (26,6);
\draw [fill=black, fill opacity = 0.1] (26,6) rectangle (27,4);

\draw (19,12) rectangle (20,8);
\node at (19.5,10) {$y_1$};
\draw (20,12) rectangle (21,9);
\node at (20.5,10.5) {$y_2$};
\draw (21,12) rectangle (22,8);
\node at (21.5,10) {$v_1$};
\draw (22,12) rectangle (23,9);
\node at (22.5,10.5) {$v_2$};

\draw (23,8) rectangle (24,6);
\node at (23.5,7) {$y_3$};
\draw (24,8) rectangle (25,6);
\node at (24.5,7) {$y_4'$};

\draw (25,8) rectangle (26,6);
\node at (25.5,7) {$v_3$};

\draw (26,6) rectangle (27,4);
\node at (26.5,5) {$y_4''$};

\end{scope}

\end{tikzpicture}
\end{center}
\caption{How the $f_\ell$ functions are
  defined and used. The bold columns of the block-diagonal
  matrix in the upper-left are associated to the columns of the target
  block-diagonal matrix on the bottom-left. The columns of the block-diagonal
  matrix in the upper-right are also associated the manner described
  in Section \ref{sec:labeling},
  although this association is not shown in the figure. We see that
  the two matrix-vector multiplications at the top can be combined
  into the single matrix product on the bottom, provided that similarly labeled
  entries of the two source matrices are guaranteed to contain
  identical values. Unlabeled cells of the middle-bottom matrix
  contain zeros. 
  Note that in the bottom-right matrix each
  segment of the original vectors $y$ and $v$ may be padded up to
  twice its original length (but not more), and it may be repeated up
  to twice and summed (as $y_4'$ and $y_4''$ are) if the columns of the associated block are mapped to two different
  blocks in the target matrix. More than two repetitions are not possible because
  the source blocks all have at most as many columns as rows.}
\label{fig:labeling}
\end{figure}
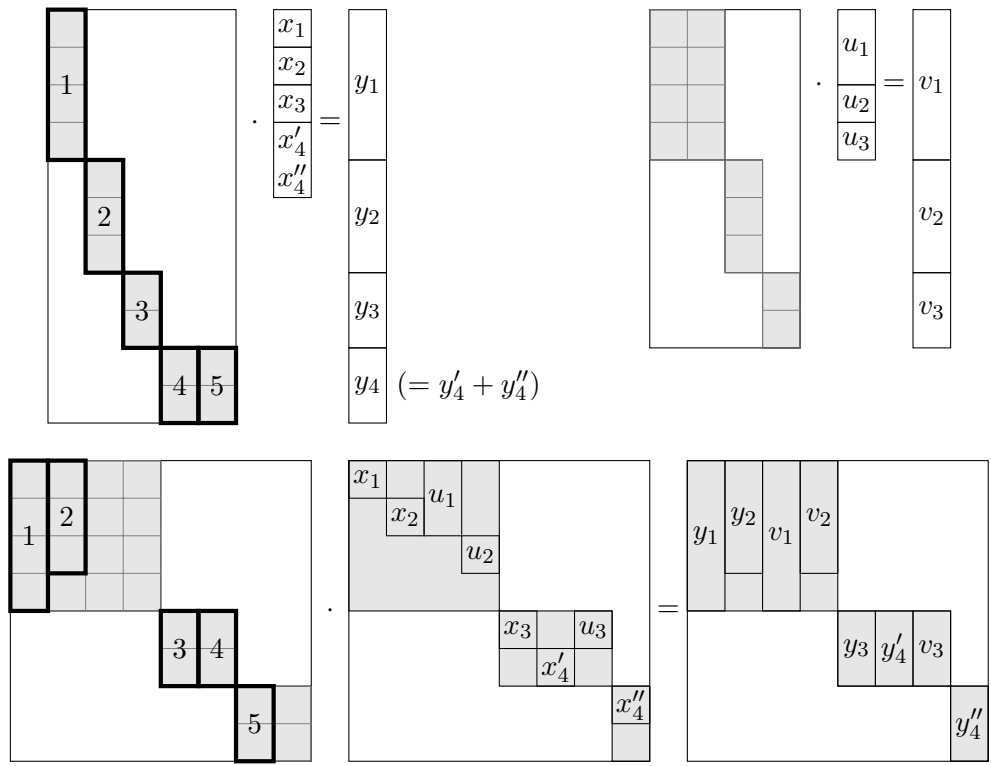

We will describe functions $f_\ell$ associating the entries of a block-diagonal
matrix of this format (which depends on $\ell$)  with a {\em target} block-diagonal matrix 
whose format is fixed as follows: 
\[\begin{array}{lcl}
2\cdot |H/N| & \mbox{ blocks of size }  & 1 \times 1 \\
 \left \lceil 2\cdot |H/N|/4 \right \rceil & \mbox{ blocks of size }  & 2 \times 2 \\
 \left \lceil 2\cdot |H/N|/16 \right \rceil & \mbox{ blocks of size }  & 4 \times 4 \\
 & \vdots  & \\
 \left \lceil 2\cdot |H/N|/2^{2i} \right \rceil & \mbox{ blocks of size }  & 2^i \times 2^i \\
 & \vdots  & \\
 2 & \mbox{ blocks of size }   & 2^{\left \lceil \log_2 |H/N| \right
                               \rceil} \times 2^{\left \lceil \log_2
                               |H/N| \right \rceil}
\end{array}\]
Note that the number of entries of this target matrix
is $O(|H/N|\log |H/N|)$, and this will be our $r$. The
asssociation specifying the map $f_\ell$ is quite simple: we take one
column at a time of the source
block-diagonal matrix, and if it has height $w$, we associate it,
top-aligned, with the next-available
column among the blocks of size $2^i \times 2^i$, for the $i$ such that $2^i/2 < w \le
2^i$. See Figure \ref{fig:labeling}. Since there can be at most $|H/N|/w <  2|H/N|/2^{i}$ columns of height
$w$ in the source matrix (which has $|H/N|$ entries in total), and the
target block-diagonal matrix has at least $2\cdot |H/N|/2^i$ columns of width
$i$, this association is possible.

We will use these mappings when applying Theorem \ref{thm:sparse} to
facilitate an efficient ``lift'' from an intermediate
representation to the final $G$-DFT. The key benefit of the mappings
is that they allow us to combine several matrix-vector products with
incompatible formats into
one, as illustrated in Figure \ref{fig:labeling}. In order to be able
to speak precisely about this combined object, we make the
following definition:
\begin{definition}[parent matrix]
\label{defn:parent}
Given a partition of $\irr(H)$ into sets $S_\ell$, matrices 
$A^\sigma$ with dimensions $\dim(\sigma)/d_\sigma\times e_\sigma $
(one for each $\sigma \in \irr(H)$),
and functions $f_\ell$ as above, satisfying
\[f_\ell(\sigma,i,j) = f_{\ell'}(\sigma', i',j') \;\;\; \Rightarrow
  \;\;\; A^{\sigma}[i,j] = A^{\sigma'}[i',j'],\]
define the {\em parent matrix} of
the $A^\sigma$ to be the matrix with the format of the
target matrix above, and with entry $(x,y)$ equal to the
value of $A^\sigma[i,j]$ if there exists $\ell$ for which $f_\ell(\sigma,
i,j) = (x, y)$, and zero otherwise.
\end{definition}
See Figure \ref{fig:parent} for an example parent matrix.
\begin{figure}[t]
\begin{center}
\begin{tikzpicture}[scale = 0.5]

\begin{scope}[shift={(-2,12)}]
\draw (1,1) rectangle (6,12);
\draw [fill=black, fill opacity = 0.1] (1,12) rectangle (2,8);
\draw[help lines] (1,12) grid (2,8);
\node at (1.5,11.5) {\small a};
\node at (1.5,10.5) {\small b};
\node at (1.5,9.5) {\small c}; 
\node at (1.5,8.5) {\small d};
\draw [fill=black, fill opacity = 0.1] (2,8) rectangle (3,5);
\draw[help lines] (2,8) grid (3,5);
\node at (2.5,7.5) {\small e};
\node at (2.5,6.5) {\small f};
\node at (2.5,5.5) {\small g}; 
\draw [fill=black, fill opacity = 0.1] (3,5) rectangle (4,3);
\draw[help lines] (3,5) grid (4,3);
\node at (3.5,4.5) {\small h};
\node at (3.5,3.5) {\small i};
\draw [fill=black, fill opacity = 0.1] (4,3) rectangle (6,1);
\draw[help lines] (4,3) grid (6,1);
\node at (4.5,2.5) {\small j};
\node at (4.5,1.5) {\small k};
\node at (5.5,2.5) {\small l};
\node at (5.5,1.5) {\small m};
\node at (3.5,0) {$A^{\sigma}$};
\end{scope}

\begin{scope}[shift={(13,15)}]
\draw (1,1) rectangle (9,9);
\draw [fill=black, fill opacity = 0.1] (1,9) rectangle (5,5);
\draw[help lines] (1,9) grid (5,5);
\node at (1.5,8.5) {\small a};
\node at (1.5,7.5) {\small b};
\node at (1.5,6.5) {\small c}; 
\node at (1.5,5.5) {\small d};
\node at (2.5,8.5) {\small e};
\node at (2.5,7.5) {\small f};
\node at (2.5,6.5) {\small g}; 
\node at (2.5,5.5) {\small n};
\node at (3.5,8.5) {\small p};
\node at (3.5,7.5) {\small q};
\node at (3.5,6.5) {\small r}; 
\draw [fill=black, fill opacity = 0.1] (5,5) rectangle (7,3);
\draw[help lines] (5,5) grid (7,3);
\node at (5.5,4.5) {\small h};
\node at (5.5,3.5) {\small i};
\node at (6.5,4.5) {\small j};
\node at (6.5,3.5) {\small k};
\draw [fill=black, fill opacity = 0.1] (7,3) rectangle (9,1);
\draw[help lines] (7,3) grid (9,1);
\node at (7.5,2.5) {\small l};
\node at (7.5,1.5) {\small m};
\node at (5,0) {parent of $\{A^{\sigma'}, A^{\sigma}\}$};
\end{scope}

\begin{scope}[shift={(6,14)}]
\draw (1,1) rectangle (5,10);
\draw [fill=black, fill opacity = 0.1] (1,10) rectangle (3,6);
\draw[help lines] (1,10) grid (3,6);
\node at (1.5,9.5) {\small a};
\node at (1.5,8.5) {\small b};
\node at (1.5,7.5) {\small c}; 
\node at (1.5,6.5) {\small d};
\node at (2.5,9.5) {\small e};
\node at (2.5,8.5) {\small f};
\node at (2.5,7.5) {\small g}; 
\node at (2.5,6.5) {\small n};
\draw [fill=black, fill opacity = 0.1] (3,6) rectangle (4,3);
\draw[help lines] (3,6) grid (4,3);
\node at (3.5,5.5) {\small p};
\node at (3.5,4.5) {\small q};
\node at (3.5,3.5) {\small r}; 
\draw [fill=black, fill opacity = 0.1] (4,3) rectangle (5,1);
\draw[help lines] (4,3) grid (5,1);
\node at (4.5,2.5) {\small h};
\node at (4.5,1.5) {\small i};
\node at (3,0) {$A^{\sigma'}$};
\end{scope}

\end{tikzpicture}
\end{center}
\caption{An example parent matrix. Unlabeled entries are zero.}
\label{fig:parent}
\end{figure}
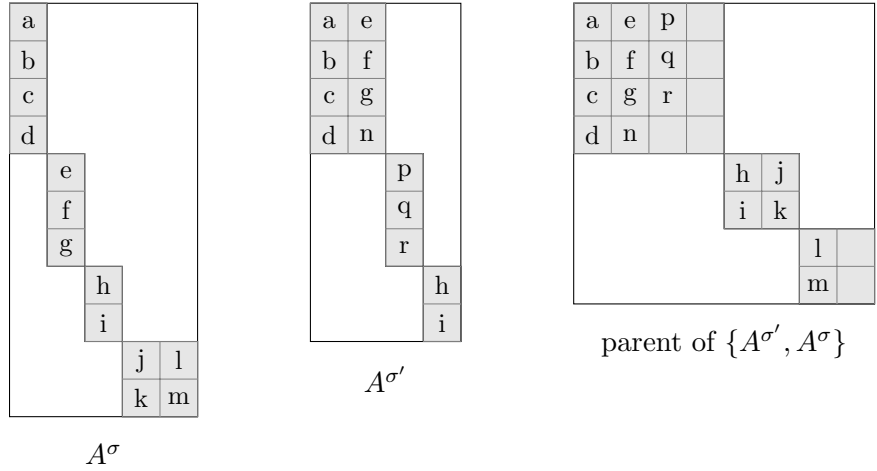

\subsubsection{Computing the intermediate representation}
We are at the point now where we can compute the intermediate
representation, which we then lift to the final $G$-DFT, making
critical use of
the just-described labellings $f_\ell$. The setup is as
follows: $H$ and $K$ are proper subgroups of group $G$, and $H \cap K = N$
is normal in $G$. Let $X$ be a system of distinct coset representatives of $N$ in $H$ and
let $Y$ be a system of distinct coset representatives of $N$ in $K$. Thus $H = XN$
and $K = NY$. Note that $HK = XNY$ with
uniqueness of expression. 

When applying the triple subgroup
reduction in our final result, it will happen that 
\[\frac{|G|}{|HK|} = \frac{|G||N|}{|H||K|}\]
is negligible, and notice that in this case, if $H$-DFTs, $K$-DFTs, and $N$-DFTS 
have nearly-linear algorithms, then indeed the cost of applying
the next lemma is nearly-linear in $|G|$ as desired. 

\begin{lemma}
\label{lem:intermediate}
With $|Y|$ many $H$-DFTs, $O(|H/N|\log |H/N|)\cdot|Y|$ many inverse $N$-DFTs, and $O(|H/N|\log|H/N|)$
many $K$-DFTs, we can compute, from $\alpha \in \C[G]$ supported on $HK$, the
following expression:
\begin{equation}
\label{eq:complicated1}
\sum_{n \in N}\sum_{y \in Y} \bigoplus_{\tau \in
  \irr(K)} P_{n, y} \otimes \tau(ny)^T
\end{equation}
where $P_{n,y}$ is the parent matrix of the matrices $\{M_{n,y}^\sigma
: \sigma \in \irr(H)\}$, and for each $\sigma, y$, the
$M_{n,y}^\sigma$ satisfy (with respect to an $N$-adapted basis for $\irr(H)$):
\begin{equation}
\label{eq:complicated2}
\sum_{n \in N} (M^{\sigma}_{n,y} \otimes J^{\sigma})\sigma(n) =
\sum_{h \in H}
\alpha_{hy}\sigma(h).
\end{equation}
where $J^{\sigma}$ is the $\dim(\sigma)/e_\sigma  \times d_\sigma$
matrix $(I_{d_\sigma} | I_{d_\sigma} | \cdots | I_{d_\sigma})$ as in
Theorem \ref{thm:sparse}.
\end{lemma}

\begin{proof}
First, compute for each $y \in Y$ and $\sigma \in \irr(H)$ the matrices
\[M_y^\sigma = \sum_{h \in H} \alpha_{hy} \sigma(h),\]
using $|Y|$ different $H$-DFTs. 
Next, apply Theorem \ref{thm:sparse}, once for each $y$, to the matrices 
\[\bigoplus_{\sigma \in  \irr(H)} M_y^\sigma \in \bigoplus_{\sigma \in
    \irr(H)} \C^{\dim(\sigma) \times \dim(\sigma)},\]
together with the labelings $f_\ell$ from Section \ref{sec:labeling}, to obtain matrices $M^{\sigma}_{n,y} \in
\C^{\dim(\sigma)/d_\sigma\times e_\sigma}$
for which 
\[\sum_{n \in N} (M^{\sigma}_{n,y} \otimes J^{\sigma})\sigma(n) = M^\sigma_y,\]
at a cost of $O(|H/N|\log |H/N|)\cdot|Y|$ many inverse $N$-DFTs. Note that these
$M^\sigma_{n,y}$ satisfy Equation (\ref{eq:complicated2}). Let
$P_{n,y}$ be the parent matrix of the matrices $\{M_{n,y}^\sigma
: \sigma \in \irr(H)\}$.

For each $(i,j)$, the vector $\beta$ with
$\beta[ny] = P_{n,y}[i,j]$ is an element of $\C[K]$ and we
perfom a $K$-DFT on it;
this entails computing at most $O(|H/N|\log |H/N|)$ different $K$-DFTs
because this is the number of entries in the blocks of the block-diagonal matrices $P_{n,y}$. At this point we hold, in the aggregate, all of the entries of Expression
(\ref{eq:complicated1}) in the statement of the lemma, and the proof
is complete. 
\end{proof}

\subsubsection{Lifting to a $G$-DFT}
\label{lift}

In this section we show how to efficiently lift the intermediate
representation, Expression \ref{eq:complicated1} 
computed via Lemma \ref{lem:intermediate},
to a $G$-DFT. We continue with the notation of the
previous section.

Let $\irr^*(H)$ denote the {\em multiset} of irreps of $H$ that occur in the
restrictions of the irreps of $G$ to $H$ (with the correct
multiplicities), and similarly let $\irr^*(K)$ denote the
{\em multiset} of irreps of $K$ that occur in the restrictions of the irreps
of $G$ to $K$. Let $S$ and $T$ be the change of basis
matrices that satisfy:
\begin{eqnarray*}
S\left ( \bigoplus_{\sigma \in \irr^*(H)} \sigma(h) \right ) S^{-1} &  = & 
\bigoplus_{\rho \in \irr(G)}\rho(h) \;\;\;\; \forall h \in H \\
T\left (\bigoplus_{\tau \in \irr^*(K)} \tau(k) \right ) T^{-1} & = &
\bigoplus_{\rho \in \irr(G)}\rho(k) \;\;\;\; \forall k \in K.
\end{eqnarray*}
We further specify that $S$ should be with respect to an $N$-adapted
basis for $\irr(H)$. 

Notice that for $n \in N = H \cap K$, we have:
\begin{equation*}
S\left ( \bigoplus_{\sigma \in \irr^*(H)} \sigma(n) \right ) S^{-1}
= T\left (\bigoplus_{\tau \in \irr^*(K)} \tau(n) \right ) T^{-1},
\end{equation*} or equivalently

\begin{equation}
\label{eq:slide}
\left ( \bigoplus_{\sigma \in \irr^*(H)} \sigma(n) \right ) S^{-1}T
= S^{-1}T\left (\bigoplus_{\tau \in \irr^*(K)} \tau(n) \right ),
\end{equation} 
a fact we will use shortly.

A $G$-DFT with input $\alpha$ supported on $HY = HK$ is the expression:
\begin{eqnarray*}
\sum_{\substack{h \in H\\ y \in Y}} \alpha_{hy}
\bigoplus_{\rho \in \irr(G)} \rho(hy) & = & 
\sum_{y \in Y} \left (\sum_{h \in H}  \alpha_{hy} \bigoplus_{\rho \in
    \irr(G)}\rho(h)\right )\cdot \left (\bigoplus_{\rho \in \irr(G)}
  \rho(y) \right )\\
& = & \sum_{y \in Y} S\left (\sum_{h \in H}  \alpha_{hy} \bigoplus_{\sigma \in
    \irr^*(H)}\sigma(h)\right )S^{-1}T\left (\bigoplus_{\tau \in
      \irr^*(K)} \tau(y) \right )T^{-1}
\end{eqnarray*}
Now for each $y \in Y$, the left-most parenthesized expression is an
$H$-DFT, with certain blocks repeated. By Equation (\ref{eq:complicated2}) in the
statement of Lemma \ref{lem:intermediate},
each such expression can be rewritten in terms of matrices $M_{n,y}^\sigma$, yielding:
\begin{eqnarray}
\sum_{\substack{h \in H\\ y \in Y}} \alpha_{hy}
\bigoplus_{\rho \in \irr(G)} \rho(hy) & = & \sum_{\substack{y \in Y\\ n \in N}} S\left (
  \bigoplus_{\sigma \in \irr^*(H)}(M_{n,y}^{\sigma} \otimes J^{\sigma})\sigma(n)\right )S^{-1}T\left (\bigoplus_{\tau \in
  \irr^*(K)} \tau(y) \right )T^{-1} \nonumber \\
& = & \sum_{\substack{y \in Y\\ n \in N}} S\underbrace{\left (
  \bigoplus_{\sigma \in \irr^*(H)}(M_{n,y}^{\sigma} \otimes J^{\sigma})\right )S^{-1}T\left (\bigoplus_{\tau \in
    \irr^*(K)} \tau(ny) \right )}_{(*)}T^{-1} \label{eq:star-sum}
\end{eqnarray} 
where the last line invoked Equation (\ref{eq:slide}) to move
$\sigma(n)$ past $S^{-1}T$.

We now focus on Expression $(*)$.  By Proposition \ref{prop:mat-mult}
we can express Expression $(*)$ as
\begin{equation}
\left (\bigoplus_{\sigma \in \irr^*(H), \tau \in \irr^*(K)}
  \left ((M_{n,y}^{\sigma} \otimes J^{\sigma}) \otimes \tau(ny)^T\right
  )\right ) \cdot \mbox{vec}(S^{-1}T)= \mbox{vec}(*). 
\label{eq:vectorized}
\end{equation}
We next apply two types of simplifications to the block-diagonal matrix
on the left. In each, we observe that 
equalities among blocks allow us to simplify that
block-diagonal matrix, at the expense of arranging portions of
$\mbox{vec}{S^{-1}T}$ and $\mbox{vec}{(*)}$ into block-diagonal matrices and summing certain entries. The first such
observation is that computing 
\[\left (\begin{array}{c|c} A \\ \hline & A \end{array} \right ) \cdot
  \left (\begin{array}{c} x_1 \\ \hline x_2 \end{array} 
\right ) =  \left (\begin{array}{c} y_1 \\ \hline y_2 \end{array} 
\right ) \]
is equivalent to computing $A \cdot (x_1 | x_2) = (y_1 | y_2)$. The second observation is that computing
\[ (A | A) \cdot \left
    ( \begin{array}{c} x_1 \\ x_2 \end{array} \right ) = y \]
is equivalent to computing $A\cdot (x_1 + x_2) = y$. 

Using the first observation we can thus
simplify Equation \ref{eq:vectorized} to:
\[\left (\bigoplus_{\sigma \in \irr(H), \tau \in \irr(K)}
  \left ((M_{n,y}^{\sigma} \otimes J^{\sigma}) \otimes \tau(ny)^T\right
  )\right )\cdot X_0 = Y_0,\]
where $X_0$ is a block-diagonal matrix whose entries coincide with the
entries of $S^{-1}T$.
Next, we notice that $J^\sigma = I_{d_\sigma} \otimes (1,1, \ldots
1)$.
The first observation then allows us to simplify Equation \ref{eq:vectorized}
futher to:
\[\left (\bigoplus_{\sigma \in \irr(H), \tau \in \irr(K)}
  \left ((M_{n,y}^{\sigma} \otimes (1,1, \ldots
1)) \otimes \tau(ny)\right
)\right )\cdot X_1 = Y_1\]
where again the entries of $X_1$ coincide with the
entries of $S^{-1}T$, 
and the second observation allows us to simplify to:
\begin{equation}
\label{eq:before-parent}
\left (\bigoplus_{\sigma \in \irr(H), \tau \in \irr(K)}
  M_{n,y}^{\sigma} \otimes \tau(ny)^T \right) \cdot X_2 = Y_2,
\end{equation}
where now $X_2$ is a block-diagonal matrix whose entries
are sums of entries of $S^{-1}T$. 

As in the statement of Lemma \ref{lem:intermediate}, for each $n, y$, let $P_{n,y}$ be the parent
matrix of the matrices $\{M_{n,y}^{\sigma} : \sigma \in \irr(H)\}$. We
can rewrite Expression (\ref{eq:before-parent}) as 
\begin{equation}
\label{eq:after-parent}
\left (\bigoplus_{\tau \in \irr(K)}
    P_{n,y}\otimes \tau(ny)^T \right ) \cdot X_3 = Y_3,
\end{equation}
where $X_3$ is a block-diagonal matrix whose entries are sums of
entries of $S^{-1}T$.

The square blocks of the 
block-diagonal matrix
\[\left (\bigoplus_{\tau \in \irr(K)}
    P_{n,y}\otimes \tau(ny)^T \right )\] have 
dimensions $a_i$ with the property that \[\sum_i a_i^2 = O(|H/N|\log
|H/N|)\cdot |K|,\]
using our earlier accounting for the block sizes of a parent matrix,
together with the fact that  $\sum_{\tau \in \irr(K)} \dim(\tau)^2 = |K|$.  
Each $a_i \times a_i$ block is multiplied by an $a_i
\times w_i$ block of $X_3$, to yield an $a_i \times w_i$ block of the
product matrix $Y_3$. We now argue that the $w_i$ satisfy $\sum_i a_iw_i =
4|G|$. Each of the two transformations applied to obtain
block-diagonal matrices $Y_0, Y_1$ and then $Y_2$ preserve the number
of entries of the result matrix; these matrices therefore have $|G|$
entries in the blocks. The final transformation results in a
block-diagonal matrix $Y_3$ which may have {\em more} entries than
$|G|$, but this number can be larger by only a factor of four, as illustrated in Figure
\ref{fig:labeling}. This is because each column of a block of $Y_2$
may need to be padded to at most twice its original length, and
repeated up to two times (and no more, because the blocks of the
$M^\sigma_{n,y}$ have no more columns than rows, and thus can spill over at most two blocks
in the parent matrix). Thus the number of entries in the blocks of
$Y_3$ which equals $\sum_i a_iw_i$, is
at most $4|G|$ as stated. 

We conclude that the block-matrix multiplication in
Expression (\ref{eq:after-parent}) can be
performed efficiently as summarized in the following lemma.

\begin{lemma}
\label{lem:lift}
The map from 
\[\sum_{n \in N}\sum_{y \in Y} \bigoplus_{\tau \in \irr(K)}
    P_{n,y}\otimes \tau(ny)^T\]
as computed from input $\alpha$ supported on $HY = HK$ in Lemma \ref{lem:intermediate},
to a $G$-DFT with, can  
 be computed at a cost of
$O(|G|^{\omega/2 + \epsilon} )$
operations, for all $\epsilon > 0$.  
\end{lemma}
\begin{proof}
We describe how to map a summand $\bigoplus_{\tau \in \irr(K)}
    P_{n,y}\otimes \tau(ny)^T$ to the corresponding summand of
    Expression (\ref{eq:star-sum}). This map will be {\em linear} and will not
    depend on $n,y$, so we apply it once to the entire sum computed by
    Lemma \ref{lem:intermediate}, to obtain
    Expression (\ref{eq:star-sum}), which is the promised $G$-DFT. 

We need to perform matrix multiplications of format $\langle a_i,a_i,
w_i\rangle$, and we know that $\sum_i a_i^2 = O(|H/N|\log
|H/N|)\cdot |K|= L$ and $\sum_i a_iw_i = 4|G|$.
The cost of such a multiplication is at most $\max(O(a_i^{\omega + \epsilon}),
O(a_i^{\omega-1 + \epsilon}w_i))$ for all $\epsilon > 0$.  Replacing
the maximum with a sum, and letting $a_{\max} = \max_i a_i$, we obtain an
upper bound on the number of operations of 
\begin{equation}
\label{eq:lift}
\sum_i O(a_i^{\omega+\epsilon}) + 
O(a_i^{\omega-1 + \epsilon}w_i) = O(a_{\max}^{\omega - 2 + \epsilon}) \sum_i
a_i^2 +a_iw_i \le L^{(\omega - 2 + \epsilon)/2}\cdot (L+4|G|).
\end{equation}
We
need to pre-multiply by $S$ and post-multiply by $T^{-1}$ to obtain a
summand of Expression (\ref{eq:star-sum}). Both $S$ and $T^{-1}$ are
block-diagonal with one block for each $\rho \in \irr(G)$, with
dimension $\dim(\rho)$. Thus the cost of this final pre- and post-
multiplication is 
\[\sum_{\rho \in \irr(G)}O(\dim(\rho)^{\omega + \epsilon})\]
which is at most $O(|G|^{\omega/2 + \epsilon})$ by Proposition \ref{prop:prelim-inequality}
with $\alpha = \omega + \epsilon$. 
The theorem follows from the fact that $|H||K|/|N| \le |G|$, and thus
Expression (\ref{eq:lift}) is also upper-bounded by $O(|G|^{\omega/2 +
  \epsilon})$ (absorbing logarithmic terms into
$|G|^{\epsilon/2}$). 
\end{proof}

We now have the main theorem putting together the entire triple
subgroup reduction:

\begin{theorem}[triple subgroup reduction]
\label{thm:triple-subgroup}
Let $G$ be a finite group and let $H,K$ be proper subgroups with $N = H \cap
K$ normal in $G$. Then we can compute generalized DFTs with respect to $G$ at
the cost of 
\begin{itemize}
\item $|K|/|N|$ many $H$-DFTs, \\
\item $O(|H||K|/|N|^2\log |H/N|)$ many
inverse $N$-DFTs, \\
\item  $O(|H/N|\log|H/N|)$ many $K$-DFTs, 
\end{itemize}
plus $O(|G|^{\omega/2 +
  \epsilon})$ operations, all repeated $O(|G|\log|G|/|HK|)$ many
times, for all $\epsilon > 0$. 
\end{theorem}
\begin{proof}
By Lemma \ref{lem:intermediate} we can compute the intermediate
representation of a $G$-DFT supported on $HK$, and applying the 
map of Lemma \ref{lem:lift} to this intermediate representation yields
a $G$-DFT supported on $HK$. ByTheorem \ref{thm:trick} we can 
compute a general $G$-DFT at the cost of repeating these two
steps $O(|G|\log|G|/|HK|)$ many
times.
\end{proof}

\subsubsection{Triple subgroup structure in finite groups}

\begin{figure}\begin{center}
\begin{tabular}{|l|l|l|l|}
\hline
Name & Family & $|W|$ & $|G|$ \\
\hline
Chevalley & $A_\ell(q)$ & $(\ell + 1)!$ & $q^{\Theta(\ell^2)}$ \\
&$B_\ell(q)$ & $2^\ell\ell!$ & $q^{\Theta(\ell^2)}$ \\
&$C_\ell(q)$ & $2^\ell\ell!$ & $q^{\Theta(\ell^2)}$ \\
&$D_\ell(q)$ & $2^{\ell-1}\ell!$ & $q^{\Theta(\ell^2)}$ \\
\hline
Exceptional & $E_6(q)$ & $O(1)$ & $q^{\Theta(1)}$ \\
Chevalley & $E_7(q)$ & $O(1)$ & $q^{\Theta(1)}$ \\
& $E_8(q)$ & $O(1)$ & $q^{\Theta(1)}$ \\
& $F_4(q)$ & $O(1)$ & $q^{\Theta(1)}$ \\
& $G_2(q)$ & $O(1)$ & $q^{\Theta(1)}$ \\
\hline
Steinberg & ${}^2A_\ell(q^2)$ & $2^{\lceil \ell/2 \rceil}\lceil \ell/2\rceil!$ & $q^{\Theta(\ell^2)}$
\\
& ${}^2D_\ell(q^2)$ & $2^{\ell-1}(\ell-1)!$ & $q^{\Theta(\ell^2)}$\\
& ${}^2E_6(q^2)$ &  $O(1)$ & $q^{\Theta(1)}$ \\
&${}^3D_4(q^3)$ &  $O(1)$ & $q^{\Theta(1)}$\\
\hline
Suzuki & ${}^2B_2(q)$, $q = 2^{2n+1}$  &  $O(1)$ & $q^{\Theta(1)}$ \\
\hline
Ree & ${}^2F_4(q)$, $q = 3^{2n+1}$  &  $O(1)$ & $q^{\Theta(1)}$ \\
& ${}^2G_2(q)$, $q = 3^{2n+1}$  &  $O(1)$ & $q^{\Theta(1)}$ \\
\hline
\end{tabular}
\end{center}
\caption{Families of finite groups $G$ of Lie type, together with
  the size of their associated Weyl group $W$. These include all
  simple finite groups other than cyclic groups, the alternating groups, the 26
  sporadic groups, and the Tits group. See \cite{lev,wiki, carter} for
  sources. The Suzuki, Steinberg and Ree families are also called the
  {\em twisted Chevalley} groups.}
\label{table:lie-type}
\end{figure}

Our main structural theorem on finite groups is the following
\begin{theorem}
\label{thm:triple-subgroup-structure}
There exists a monotone increasing function $f(x) \le 2^{c\sqrt{\log
    x}\log\log x}$ for a universal constant $c \ge 1$, such that, for
every nontrivial finite group $G$ one of the following holds
\begin{enumerate}
\item $G$ has a (possibly trivial) normal subgroup $N$ and $G/N$ is cyclic of
  prime order, or
\item $G$ has a (possibly trivial) normal subgroup $N$ and $G/N$ has proper
  subgroups $X, Y$ with $X \cap Y = \{1\}$ and for which $|X||N||Y| \ge |G|/f(G)$.
\end{enumerate}
\end{theorem}
To connect this theorem to our usage in the previous sections, think of
$H$ as being the subgroup $\overline{X}N$ and $K$ as being the
subgroup $N\overline{Y}$, where $\overline{X}$ and $\overline{Y}$ are
lifts of $X$ and $Y$, respectively, from $G/N$ to $G$. 

\begin{proof}
Let $N$ be a maximal normal subgroup of $G$. Then $G/N$ is simple. If
it is cyclic of prime order, then we are done. Otherwise we have the
following cases, by the Classification Theorem:
\begin{enumerate}
\item $G/N$ is an alternating group $A_n$ for $n \ge 5$. In this case,
  let $X$ be the subgroup of $G/N$ isomorphic to $A_{n-1}$ and $Y$
  the trivial subgroup of $G/N$. 
\item $G/N$ is a finite group of Lie Type. In this case, we refer to
  Table \ref{table:lie-type}, and we have the
  following description from Carter \cite{carter}. For Chevalley and
  exceptional Chevalley groups, we have that there are subgroups $B$
  and $U_w^{-}$ (for each $w$ in the associated Weyl group $W$) so that elements of $G/N$ can be
  expressed {\em uniquely} as $bn_wu_w$, where $b \in B$, $n_w$ is a
  lift of $ w\in W$ to $G$, and $u_w \in U_w^{-}$ (see Corollary
  8.4.4 in Carter \cite{carter}). Uniqueness implies that the conjugate subgroup
  $n_wU_w^{-}n_w^{-1}$ has trivial intersection with $B$; also,
  by an averaging argument, there exists $w \in W$ for which
  $|Bn_wU_w^{-}n_w^{-1}| \ge |G/N|/|W|$. We take $X = B$ and $Y =
  n_wU_w^{-}n_w^{-1}$. For twisted Chevalley groups, we have an
  identical situation (see Corollary
  13.5.3 in Carter \cite{carter}), with subgroup $B$ replaced by $B^1$ and
  subgroup $U_w^{-}$ replaced by $(U_w^{-})^{1}$ (in Carter's
  notation).  Again by an averaging argument there exists $w \in W$ for which
  $|B^{1}n_w(U_w^{-})^{1}n_w^{-1}| \ge |G/N|/|W|$, and subgroups
  $B^{1}$ and $n_w(U_w^{-})^{1}n_w^{-1}$ have trivial intersection; so
  we take them as our $X$ and $Y$, respectively. Finally we verify from Table
  \ref{table:lie-type} that in all cases we have $f(|G/N|) \ge
  |W|$. Thus \[|X||N||Y| \ge |N||G/N|/|W| \ge |N||G/N|/f(|G/N|) \ge
    |G|/f(|G|)\]
where we used the fact that $f$ is increasing.

\item $G/N$ is a one of the 26 sporadic groups or the Tits group. In
  this case, we can take $X = Y = \{1\}$, by choosing $c$ in the
  definition of $f(x)$ sufficiently large. 
\end{enumerate}
\end{proof}

\subsubsection{Putting it together}

Using the structural theorem and the new triple-subgroup reduction recursively, we obtain our final result:

\begin{theorem}[main]
For any finite group $G$, there is an arithmetic algorithm computing
generalized DFTs with respect to $G$, using $O(|G|^{\omega/2 +
  \epsilon})$ operations, for any $\epsilon > 0$. 
\end{theorem}
\begin{proof}
Fix an arbitrary $\epsilon > 0$. Consider the following recursive
algorithm to compute a $G$-DFT. If $G$ is trivial then computing a
$G$-DFT is as well. If $G$ has a proper
subgroup $H$ of order larger than $|G|^{1 - \epsilon/2}$ then we apply Theorem \ref{thm:single-subgroup} to
compute a $G$-DFT via several $H$-DFTs. Otherwise, applying
Theorem \ref{thm:triple-subgroup-structure}, we obtain a (possibly
trivial) normal subgroup $N$, and two proper subgroups of $G$, $H$ and
$K$, with $N = H \cap K$. If $G/N$ is cyclic of prime order, we apply
Corollary \ref{cor:prime-index} to compute a $G$-DFT via several
$N$-DFTs. Otherwise, we apply Theorem \ref{thm:triple-subgroup} to
compute a $G$-DFT via several $H$-DFTs, $K$-DFTS, and inverse
$N$-DFTs. 

Let $T(n)$ denote an upper bound on the operation count of this
recursive algorithm for any group of order $n$. We will prove by
induction on $n$, that there is a universal constant $C_\epsilon$ for which 
\[T(n) \le C_\epsilon n^{\omega/2 + \epsilon}\log n.\]

In the case that we apply Theorem \ref{thm:single-subgroup}, the cost
is the cost of $[G:H]$ many $H$-DFTs plus $A_0[G:H]|G|^{\omega/2 +
  \epsilon/2}$ operations (where $A_0$ is the constant hidden in the
big-oh), and by induction this is at most:
\[C_\epsilon [G:H]|H|^{\omega/2 + \epsilon} \log|H| + A_0[G:H]|G|^{\omega/2 +
  \epsilon/2} \le C_\epsilon |G|^{\omega/2 + \epsilon} (\log
|G| - 1) + A_0|G|^{\omega/2 +
  \epsilon}  \]
which is indeed less than $C_\epsilon |G|^{\omega/2 + \epsilon}\log
|G|$ provided $C_\epsilon \ge A_0$.

In the case that we apply Corollary \ref{cor:prime-index}, our cost is
$p$ many $N$-DFTs, plus $A_1|G|^{\omega/2 + \epsilon})$ operations,
which by induction is at most 
\[C_\epsilon p(|G|/p)^{\omega/2 + \epsilon} \log(|G|/p) + A_1|G|^{\omega/2 +
  \epsilon} \le  C_\epsilon|G|^{\omega/2 + \epsilon}(\log
|G|-1) + A_1|G|^{\omega/2 +
  \epsilon},\]
which is indeed less than $C_\epsilon |G|^{\omega/2 + \epsilon}\log
|G|$ provided $C_\epsilon \ge A_1$.

Finally, in the case that we apply Theorem \ref{thm:triple-subgroup},
let $A_2$ be the maximum of the constants hidden in the big-ohs in the
statement of the Theorem (applied with $\epsilon/2$). Note that by selecting $C_\epsilon$ sufficiently large, we may
assume that $G$ is sufficiently large, so that two inequalities hold:
\begin{eqnarray*}
  A_2 |H/N| \log |H/N| & \le &  \frac{|H/N|^{\omega/2 +
                               \epsilon}}{4A_2f(|G|)\log |G|}\\
  |K/N| & \le &  \frac{|K/N|^{\omega/2 + \epsilon}}{4A_2f(|G|)\log|G|}\\
\end{eqnarray*}
and this is possible because Theorem
\ref{thm:triple-subgroup-structure} implies that $|H/N|$ (resp. $|K/N|$) 
are at least $|G|^{\epsilon/2}/f(|G|)$, as otherwise $|K|$ (resp. $|H|$) would exceed $|G|^{1 -
  \epsilon/2}$. Our cost is
$|K/N|$ many $H$-DFTs, $A_2|H||K|/|N|^2\log |H/N|$ many inverse
$N$-DFTs, $A_2|H/N|\log |H/N|$ many $K$-DFTs, plus
$A_2|G|^{\omega/2 + \epsilon/2}$ operations, all repeated
$A_2|G|\log|G|/|HK| \le A_2 f(|G|)\log|G|$ times. By induction, this is at most 
\begin{eqnarray*}
\left (C_\epsilon|K/N| |H|^{\omega/2 + \epsilon}\log |H| + 
C_\epsilon A_2|H||K|/|N|^2 \log |H/N| |N|^{\omega/2 +   
  \epsilon}\log |N| \right . & &  \\
\left . C_\epsilon A_2|H/N|\log |H/N| |K|^{\omega/2 +
  \epsilon}\log |K| +  
  A_2|G|^{\omega/2 + \epsilon/2} \right )\cdot A_2f(|G|)\log|G|
\end{eqnarray*}
Now the first three summands are each at most
\[\frac{C_\epsilon |G|^{\omega/2 + \epsilon}\log |G|}{4A_2f(|G|)\log
  |G|}\]
as is the fourth summand provided $|G|$ is sufficiently large. Thus the entire expression
is at most $C_\epsilon|G|^{\omega/2 + \epsilon}\log |G|$,
as required. This completes the proof.
\end{proof}

\section{Open problems}

Is there a proof of Theorem \ref{thm:triple-subgroup-structure} that does not need the
Classification Theorem? 
A second question is whether the dependence on $\omega$ can be removed. Alternatively, can
one show that a running time that depends on $\omega$ is necessary by
showing that an exponent-one DFT for a certain family of groups would
imply $\omega = 2$?

\paragraph{Acknowledgements.} We thank Jonah Blasiak, Tom Church, and
Henry Cohn 
for useful discussions during an AIM SQuaRE meeting. 

\bibliographystyle{alpha}

\end{document}